\documentclass[10pt]{article}

\usepackage{amsxtra,amssymb,amsthm,amsmath,latexsym}
\usepackage{appendix}
\usepackage{graphicx}
\newtheorem{thm}{Theorem}[section]

\newtheorem{lem}[thm]{Lemma}
 \newcommand{\thmref}[1]{Theorem~\ref{#1}}
 \newcommand{\lemref}[1]{Lemma~\ref{#1}}

\newcommand{\R}{{\mathbb R}}
\newcommand{\C}{{\mathbb C}}

\newcommand{\bee}{\begin{equation*}}
\newcommand{\eee}{\end{equation*}}
\newcommand{\be}{\begin{equation}}
\newcommand{\ee}{\end{equation}}
\newcommand{\pn}{\par\noindent}
\title{Uniqueness theorem for inverse scattering problem with
non-overdetermined data}
\author{A G Ramm\\
\small Department of Mathematics\\[-0.8ex]
\small Kansas State University, Manhattan, KS 66506-2602, USA\\[-0.8ex]
\small \texttt{ramm@math.ksu.edu}\\
}

\begin{document}
\date{}
\maketitle
\begin{abstract}
Let $q(x)$ be real-valued compactly supported sufficiently smooth
function. It is proved that the scattering data $A(-\beta,\beta,k)$
$\forall \beta\in S^2$, $\forall k>0$ determine $q$ uniquely.
\end{abstract}
\pn{
 }

\section{Introduction}
Since the forties of the last century the physicists were interested
in the uniqueness of the determination of a physical system by its
$S$-matrix. If the physical system is described by a Hamiltonian of
the type $H=-\nabla^2+q(x)$, then the $S$-matric is in one-to-one
correspondence with the scattering amplitude $A$,
$S=I+\frac{ik}{2\pi}A$, where $I$ is the identity operator and $A$
is an operator in $L^2(S^2)$ with the kernel $A(\beta,\alpha, k)$,
$S^2$ is the unit sphere in $\R^3$, and $k^2$ is the energy, $k>0$.
The scattering amplitude is defined as follows. If the incident
plane wave $u_0=e^{ik\alpha \cdot x}$, $\alpha\in S^2$, is scattered
by the potential $q$, then the scattering solution $u(x,\alpha,k)$
solves the scattering problem: \be\label{e1}
[\nabla^2+k^2-q(x)]u=0\quad in\quad \R^3\ee \be\label{e2}
u=e^{ik\alpha\cdot
x}+A(\beta,\alpha,k)\frac{e^{ikr}}{r}+o\left(\frac{1}{r}\right),\quad
r:=|x|\to \infty,\ \beta:=\frac{x}{r}. \ee The coefficient
$A(\beta,\alpha,k)$ is called the scattering amplitude. The problem
of interest is to determine $q(x)$ given the scattering data. This
problem is called the inverse scattering problem. The function
$A(\beta,\alpha,k)$ depends on two unit vectors $\beta$, $\alpha$,
and on the scalar $k$, i.e., on five variables.

{\it Assumption A): We assume that $q$ is compactly supported, i.e.,
$q(x)=0$ for $|x|>a$, where $a>0$ is an arbitrary large fixed
number, $q(x)$ is real-valued, i.e., $q=\overline{q}$, and $q(x)\in
H_0^\ell(B_a)$, $\ell>2$.}

Here $B_a$ is the ball centered at the
origin and of radius $a$, and $H_0^\ell(B_a)$ is the closure of
$C_0^\infty(B_a)$ in the norm of the Sobolev space $H^\ell(B_a)$ of
functions whose derivatives up to the order $\ell$ belong to
$L^2(B_a).$ It was proved in \cite{R6} (see also \cite{R35} and
\cite{R486}, Chapter 6) that if $q=\overline{q}$ and $q\in L^2(B_a)$
is compactly supported, then the resolvent kernel $G(x,y,k)$ of the
Schr\"{o}dinger operator $-\nabla^2+q(x)-k^2$ is a meromorphic
function of $k$ on the whole complex plane $k$, analytic in Im$k\geq
0$ except, possibly, of a finitely many simple poles at the points
$ik_j$, $k_j>0$, $1\leq j\leq n$, where $-k_j^2$ are negative
eigenvalues of the selftadjoint operator $-\nabla^2+q(x)$ in
$L^2(\R^3)$. Consequently, the scattering amplitude
$A(\beta,\alpha,k)$, corresponding to the above $q$, is a restriction
to the positive semiaxis $k\in[0,\infty)$ of a meromorphic on the
whole complex $k$-plane function.

It was known long ago that the
scattering data $A(\beta,\alpha,k)$ $\forall \alpha,\beta\in S^2$,
$\forall k>0$,  determine $q(x)$ uniquely. For even larger class
of potentials this result was proved in \cite{Sa}.

The above scattering data depend on five variables (two unit vectors 
$\alpha$ and $\beta$, and one scalar $k$). The potential $q(x)$
depends on three variables ($x_1, x_2, x_3$). Therefore, 
the inverse scattering problem, which consists of finding $q$
from the above scattering data,  is overdetermined.

It was proved by
the author (\cite{R228}), that the {\it fixed-energy scattering data}
$A(\beta,\alpha):=A(\beta,\alpha,k_0)$, $k_0=const>0$, $\forall
\beta \in S_1^2,$ $\forall \alpha\in S_2^2$, determine real-valued
compactly supported $q\in L^2(B_a)$ uniquely. Here $S_j^2,$ $j=1,2$,
are arbitrary small open subsets of $S^2$ (solid angles).

The scattering data $A(\beta,\alpha)$ depend on four variables (two
unit vectors), while the unknown $q(x)$ depends on three variables.
In this sense the inverse scattering problem, which consists of finding 
$q$ from the fixed-energy scattering data $A(\beta,\alpha)$, is still 
overdetermined.

For many decades there were no uniqueness theorems for 3D inverse
scattering problems with non-overdetermined data. The goal of this
paper is to prove such a theorem.

\begin{thm}\label{thm1}
If $\overline{q}=q\in H_0^\ell(B_a)$, $\ell>2$, then the data
$A(-\beta,\beta,k)$ $\forall \beta\in S^2$, $\forall k>0$, determine
$q$ uniquely.
\end{thm}
{\bf Remark 1.} The conclusion of \thmref{thm1} remains valid
if the data $A(-\beta,\beta,k)$ are known $\forall \beta\in S_1^2$
and $k\in(k_0,k_1)$ where $(k_0,k_1)\subset[0,\infty)$ is an
arbitrary small interval,  $k_1>k_0$, and $S_1^2$ is an arbitrary
small open subset of $S^2.$ 

In Section 2 we formulate some known
auxiliary results. 

In Section 3 proof of \thmref{thm1} is given.

In the Appendix a technical estimate is proved.

A brief announcement of the result, stated in Theorem 1.1, is
given in \cite{R571}.

\section{Auxiliary results}
Let \be\label{e3} F(g):=\tilde{g}(\xi)=\int_{\R^3}g(x)e^{i\xi\cdot
x} dx,\quad g(x)=\frac{1}{(2\pi)^3}\int_{\R^3} e^{-i\xi\cdot
x}\tilde{g}(\xi) d\xi. \ee

If $f*g:=\int_{\R^3}f(x-y)g(y)dy,$ then \be\label{e4}
F(f*g)=\tilde{f}(\xi)\tilde{g}(\xi),\quad
F(f(x)g(x))=\frac{1}{(2\pi)^3}\tilde{f}*\tilde{g}. \ee If
\be\label{e5} G(x-y,k):=\frac{e^{ik[|x-y|-\beta\cdot(x-y)]}
}{4\pi|x-y|},\ee 
then 
\be\label{e6}
F(G(x,k))=\frac{1}{\xi^2-2k\beta\cdot \xi}, \qquad \xi^2:=\xi\cdot
\xi. 
\ee 
The scattering solution $u=u(x,\alpha,k)$ solves (uniquely)
the integral equation \be\label{e7} u(x,\alpha,k)=e^{ik\alpha\cdot
x}-\int_{B_a}g(x,y,k)q(y)u(y,\alpha,k)dy, \ee where \be\label{e8}
g(x,y,k):=\frac{e^{ik|x-y|}}{4\pi|x-y|}.\ee If \be\label{e9}
v=e^{-ik\alpha\cdot x}u(x,\alpha,k),\ee then \be\label{e10}
v=1-\int_{B_a}G(x-y,k)q(y)v(y,\alpha,k)dy,\ee where $G$ is defined
in \eqref{e5}.

Define $\epsilon$ by the formula \be\label{e11} v=1+\epsilon.\ee
Then \eqref{e10} can be rewritten as 
\be\label{e12}
\epsilon(x,\alpha,k)=-\int_{\R^3}G(x-y,k)q(y)dy- T\epsilon,
\ee
where 
$$T\epsilon:=\int_{B_a}G(x-y,k)q(y)\epsilon(y,\alpha,k)dy.$$
Fourier transform of \eqref{e12} yields (see \eqref{e4},\eqref{e6}):
\be\label{e13}
\tilde{\epsilon}(\xi,\alpha,k)=-\frac{\tilde{q}(\xi)}{\xi^2-2k\alpha\cdot
\xi}-\frac{1}{(2\pi)^3}\frac{1}{\xi^2-2k\alpha\cdot
\xi}\tilde{q}*\tilde{\epsilon}. \ee An essential ingredient of our
proof in Section 3 is the following lemma, proved by the author
(see its proof in  \cite{R470}, p.262, or in  \cite{R425}):

\begin{lem}\label{lem1}
If $A_j(\beta,\alpha,k)$ is the scattering amplitude corresponding
to potential $q_j$, then \be\label{e14}
-4\pi[A_1(\beta,\alpha,k)-A_2(\beta,\alpha,k)]=\int_{B_1}[q_1(x)-q_2(x)]u_a(x,\alpha,k)u_2(x,-\beta,k)dx,
\ee where $u_j$ is the scattering solution corresponding to $q_j$.
\end{lem}
Consider an algebraic variety $\mathcal{M}$ in $\C^3$ defined by the
equation \be\label{e15} \theta\cdot\theta=1,\quad
\theta\cdot\theta:=\theta_1^2+\theta_2^2+\theta_3^2,\quad \theta_j\in
\C. \ee This is a non-compact variety, intersecting $\R^3$ over the
unit sphere $S^2$.

Let $R_+=[0,\infty).$ The following result is proved in \cite{R190},
p.62.
\begin{lem}\label{lem2}
If Assumption A) holds, then the scattering amplitude
$A(\beta,\alpha,k)$ is a restriction to $S^2\times S^2\times R_+$ of
a function $A(\theta',\theta,k)$ on
$\mathcal{M}\times\mathcal{M}\times\C$, analytic on
$\mathcal{M}\times\mathcal{M}$ and meromorphic on $\C$, $\theta'$,
$\theta\in \mathcal{M}$, $k\in \C$.
\end{lem}
The scattering solution $u(x,\alpha,k)$ is a meromorphic function of
$k$ in $\C$, analytic in Im$k\geq 0$, except, possibly at the points
$k=ik_j$, $1\leq j\leq n$.

We need the notion of the Radon transform
(see, e.g., \cite{R348}): 
\be\label{e16}
\hat{f}(\beta,\lambda):=\int_{\beta\cdot x=\lambda}f(x)d\sigma, \ee
where $d\sigma$ is the element of the area of the plane $\beta\cdot
x=\lambda$, $\beta\in S^2$, $\lambda=const.$ One has (see
\cite{R348}, pp. 12, 15) \be\label{e17}
\int_{B_a}f(x)dx=\int_{-a}^a\hat{f}(\beta,\lambda)d\lambda, \ee
\be\label{e18} \int_{B_a}e^{ik\beta\cdot x}f(x)dx=\int_{-a}^a
e^{ik\lambda}\hat{f}(\beta,\lambda) d\lambda, \ee \be\label{e19}
\hat{f}(\beta,\lambda)=\hat{f}(-\beta,-\lambda). \ee Finally, we
need a Phragmen-Lindel\"{o}f lemma, which is proved in  \cite{Le}, 
p.69,  and  in \cite{PS}.
\begin{lem}\label{lem3}
Let $f(z)$ be holomorphic inside an angle $\mathcal{A}$ of opening
$<\pi$;  $|f(z)|\leq c_1e^{c_2|z|}, \quad z\in \mathcal{A} $; $|f(z)|\leq
M$ on the boundary of $\mathcal{A}$; and $f$ is continuous up to the 
boundary
of $\mathcal{A}$. Then $|f(z)|\leq M,\quad \forall z\in \mathcal{A}.$
\end{lem}

\section{Proof of \thmref{thm1}}
We start with the observation that the scattering data in Remark 1
determine uniquely the scattering data in \thmref{thm1} by
\lemref{lem2}.

{\it Let us outline the ideas of the proof of \thmref{thm1}.}

Assume that
$q_j$, $j=1,2$, generate the same scattering data:
$$A_1(-\beta,\beta,k)=A_2(-\beta,\beta,k)\qquad \forall \beta\in S^2,
\quad \forall k>0,$$ and let
$$p(x):=q_1(x)-q_2(x).$$
Then by \lemref{lem1},
see equation \eqref{e14}, one gets \be\label{e20}
0=\int_{B_a}p(x)u_a(x,\beta,k)u_2(x,\beta,k)dx, \qquad \forall \beta\in
S^2,\ \forall k>0.\ee 
By \eqref{e9} and \eqref{e11} one can rewrite \eqref{e20} as 
\be\label{e21} 0=\int_{B_a}e^{2ik\beta\cdot
x}[1+\epsilon(x,k)]p(x)dx=0\quad \forall \beta\in S^2,\ \forall k>0,
\ee where \bee
\epsilon(x,k):=\epsilon:=\epsilon_1(x,k)+\epsilon_2(x,k)+
\epsilon_1(x,k)\epsilon_2(x,k).
\eee 
By \lemref{lem2} the relations \eqref{e20} and \eqref{e21} hold for 
complex $k$, 
\be\label{e22} k=\frac{\kappa+i\eta}{2},\qquad \kappa+i\eta\neq
2ik_j,\quad \eta\geq 0, \ee in particular, for $\eta>k_n,$
$\kappa\in \R$. Using formulas \eqref{e3}-\eqref{e4}, one derives
from \eqref{e21} the relation \be\label{e23}
\tilde{p}((\kappa+i\eta)\beta)+\frac{1}{(2\pi)^3}\tilde{\epsilon}*\tilde{p}=0
\qquad \forall \beta,\ \forall \kappa\in \R,\quad \eta>k_n. \ee 
One has: 
\be\label{e24} \sup_{\beta\in
S^2}|\tilde{\epsilon}*\tilde{p}|:=\sup_{\beta\in
S^2}|\int_{\R^3}\tilde{\epsilon}((\kappa+i\eta)\beta-s)\tilde{p}(s)ds|\leq
\nu(\kappa, \eta)\sup_{s\in \R^3}|\tilde{p}(s)|, \ee 
where 
$$\nu(\kappa, \eta):=\sup_{\beta\in S^2}\int_{\R^3}|\tilde{\epsilon}
((\kappa+i\eta)\beta-s)|ds.$$
We will prove that if  $\eta=\eta(\kappa)=O(\ln \kappa)$,  
then the following inequality holds:
\be\label{e25}
0<\nu(\kappa,\eta(\kappa))<1, \qquad \kappa\to \infty. \ee 
If one proves that
\be\label{e26} \sup_{\beta\in
S^2}|\tilde{p}((\kappa+i\eta(\kappa))\beta)|\geq \sup_{s\in
\R^3}|\tilde{p}(s)|,\quad \kappa\to \infty, \ee 
then it follows from
\eqref{e23}-\eqref{e26} that $\tilde{p}(s)=0$, so $p(x)=0$, and
\thmref{thm1} is proved.

{\it This completes the outline of the proof of \thmref{thm1}}.

Let us now establish estimates \eqref{e25} and \eqref{e26}. 

We assume that $p(x)\not\equiv 0$, because
otherwise there is nothing to prove. If $p(x)\not\equiv 0$, then
$$\max_{s\in\R^3}|\tilde{p}(s)|:=\mathcal{P}\neq 0.$$

\begin{lem}\label{lem4}
If Assumption A) holds, then
\be\label{e27} \lim_{\eta\to \infty}
\max_{\beta\in S^2}|\tilde{p}((\kappa+i\eta)\beta)|=\infty. \ee
For any $\kappa>0$ there is a $\eta=\eta(\kappa)$, such that
\be\label{e28}
\max_{\beta\in S^2}|\tilde{p}((\kappa+i\eta(\kappa))\beta)|=
\mathcal{P}, \ee
and
\be\label{e29} \eta(\kappa)=O(\ln\kappa)\text{\quad as \quad }\kappa\to
+\infty. \ee
\end{lem}

{\it Proof of \lemref{lem4}}.    By formula \eqref{e18} one
gets \be\label{e30}
\tilde{p}((\kappa+i\eta)\beta)=\int_{B_a}p(x)e^{i(\kappa+i\eta)\beta\cdot
x}dx=\int_{-a}^ae^{i\kappa \lambda
-\eta\lambda}\hat{p}(\beta,\lambda)d\lambda. \ee The function
$\hat{p}(\beta,\lambda)$ satisfies \eqref{e19}. Therefore
\be\label{e31} \max_{\beta\in
S^2}|\tilde{p}((\kappa+i\eta(\kappa))\beta)|=\max_{\beta\in
S^2}|\tilde{p}((\kappa-i\eta(\kappa))\beta)|. \ee 
Indeed,
\be\begin{split}\label{e32} \max_{\beta\in
S^2}|\tilde{p}((\kappa+i\eta(\kappa))\beta)|&=\max_{\beta\in
S^2}\left|\int_{-a}^ae^{i\kappa\lambda-\eta\lambda}\hat{p}(\beta,\lambda)d\lambda\right|\\
&=\max_{\beta\in
S^2}\left|\int_{-a}^ae^{-i\kappa \mu+\eta\mu}\hat{p}(\beta,-\mu)d\mu\right|\\
&=\max_{\tilde{\beta}\in S^2}\left| \int_{-a}^a e^{-i\kappa
\mu+\eta\mu}\hat{p}(-\tilde{\beta},-\mu)d\mu
\right|\\
&=\max_{\tilde{\beta}\in S^2}\left| \int_{-a}^a e^{-i\kappa
\mu+\eta\mu}\hat{p}(\tilde{\beta},\mu)d\mu
\right|\\
&=\max_{\beta\in S^2}|\tilde{p}((\kappa-i\eta)\beta)|.\\
\end{split}\ee Here we took into account that
$\hat{p}(\beta,\lambda)$ is a real-valued function, because $q_j(x)$
are real-valued. If $p(x)\not\equiv 0$, then \eqref{e30} and \eqref{e31}
imply \eqref{e27}, as follows from \lemref{lem3}. Let us give a detailed 
argument.

Consider
the function $h$ of the complex variable  $z:=\kappa+i\eta:$ 
\be\label{e33}
h:=h(z,\beta):=\int_{-a}^ae^{iz\lambda}\hat{p}(\beta,\lambda)d\lambda.\ee
If \eqref{e27} is false, then \be\label{e34} |h(z,\beta)|\leq c\quad
\forall z=\kappa+i\eta,\quad \eta\geq 0,\quad \forall \beta\in S^2,
\ee where $\kappa\geq 0$ is an arbitrary  fixed number and the constant 
$c>0$ does not depend on $\beta$ and $\eta$. 

Thus, $|h|$ is bounded on the ray $\{\kappa=0, \eta\geq 0\}$, which is 
part of the boundary of the right angle $\mathcal{A}$, and the other
part of its boundary is the ray  $\{\kappa\geq 0, \eta= 0\}$. Let us check 
that $|h|$ is bounded on this ray also.

One has
\be\label{e35} |h(\kappa,\beta)|\leq
|\int_{-a}^ae^{i\kappa \lambda}\hat{p}(\beta,\lambda)d\lambda|\leq
\int_{-a}^a|\hat{p}(\beta,\lambda)|d\lambda\leq c, \ee where $c$
stands for {\it various} constants. 
From \eqref{e34}-\eqref{e35} it
follows that on the boundary of the right angle $\mathcal{A}$, namely, 
on the two
rays $\{\kappa\geq 0, \eta=0\}$ and $\{\kappa=0, \eta\geq 0,\}$, the
entire function $h(z,\beta)$ is bounded, $|h(z,\beta)|\leq c$,
and inside $\mathcal{A}$ this function  satisfies 
the estimate
\be\label{e36} |h(z,\beta)|\leq
e^{|\eta|a}\int_{-a}^a|\hat{p}(\beta,\lambda)|d\lambda\leq
ce^{|\eta|a}. \ee
Therefore, by \lemref{lem3}, $|h(z,\beta)|\leq c$ in the
whole angle $\mathcal{A}$. 

By \eqref{e31} the same argument is applicable to the remaining three
right angles, the union of which is the whole complex $z-$plane $\C$. 
Therefore 
\be\label{e37} \sup_{z\in
\C,\beta\in S^2}|h(z,\beta)|\leq c. \ee 
This implies that
$h(z,\beta)=c.$ 

Since $\hat{p}(\beta,\lambda)\in L^1(-a,a)$, the
relation 
\be\label{e38}
\int_{-a}^ae^{iz\lambda}\hat{p}(\beta,\lambda)d\lambda=c\quad \forall
z\in \C, 
\ee implies that $c=0$, so $\hat{p}(\beta,\lambda)=0.$
Therefore
$p(x)=0$, contrary to our assumption. Consequently, the relation
\eqref{e27} is
proved. \hfill $\Box$\\
Let us derive estimate \eqref{e29}.

From the assumption $p(x)\in H^{\ell}_0(B_a)$ it
follows that \be\label{e39} |\tilde{p}((\kappa+i\eta)\beta)|\leq
c\frac{e^{a|\eta|}}{(1+\kappa^2+\eta^2)^{\ell/2}}. \ee
This inequality is established in Lemma 3.2, below.

The right-hand side of this inequality is of the order $O(1)$ as
$\kappa\to
\infty$ if and only if $\eta=O(\ln \kappa)$, which is the relation
\eqref{e29}.\hfill $\Box$

\begin{lem}\label{lem5}
If $p\in H_0^\ell(B_a)$ then estimate \eqref{e39} holds.
\end{lem}
\begin{proof}
Consider $\partial_jp:=\frac{\partial p}{\partial x_j}.$ One has
\be\begin{split}
\left|\int_{B_a}\partial_jpe^{i(\kappa+i\eta)\beta\cdot x}
dx\right|&=\left|-i(\kappa+i\eta)\beta_j\int_{B_a}p(x)e^{i(\kappa+i\eta)\beta\cdot
x} dx\right|\\
&\leq (\kappa^2+\eta^2)^{1/2}|\hat{p}((\kappa+i\eta)\beta)|\\
&\leq (1+(\kappa^2+\eta^2)^{1/2})|\tilde{p}((\kappa+i\eta)\beta)|.
\end{split}\ee
Therefore
\be\label{e36'} |\tilde{p}((\kappa+i\eta)\beta)|\leq
c[1+(\kappa^2+\eta^2)]^{-1}e^{|\eta|a}.
\ee
Repeating this argument
one gets estimate \eqref{e39}. \hfill $\Box$

Estimate \eqref{e36'} implies that if estimate \eqref{e29} holds
 and $\kappa \to \infty$, then  the quantity
$\sup_{\beta \in S^2}|\tilde{p}((\kappa+i\eta)\beta)|$
remains bounded as $\kappa\to \infty$.

If $\eta$ is fixed and $\kappa \to \infty$, then $\sup_{\beta \in
S^2}|\tilde{p}((\kappa+i\eta)\beta)| \to 0$ by the Riemann-Lebesgue
lemma. This and  \eqref{e27} imply the existence of $\eta=\eta(\kappa)$,
such that  \eqref{e28} holds, and, consequently,  \eqref{e26} holds. 
This  $\eta(\kappa)$ satisfies
 \eqref{e29} because $\mathcal{P}$ is bounded.
\end{proof}

To complete the
proof one has to establish estimate \eqref{e25}. This
estimate will be established if one proves the following:
\be\label{e40} \lim_{\kappa\to \infty}\nu(\kappa):=\lim_{\kappa\to
\infty}\nu(\kappa,\eta(\kappa)) =0,
\ee 
where $\eta(\kappa)=O(\ln\kappa)$ and 
\be\label{e41} \nu(\kappa,\eta)=\sup_{\beta\in
S^2}\int_{\R^3}|\tilde{\epsilon}((\kappa+i\eta)\beta-s)|ds. \ee

Our argument is valid for $\epsilon_1$, $\epsilon_2$ and
$\epsilon_1\epsilon_2$, so we will use the letter $\epsilon$ and
equation \eqref{e13} for $\tilde{\epsilon}$. 

It is sufficient to check estimates \eqref{e40}-\eqref{e41} for the 
function $\tilde{q}(\xi)(\xi^2-2k\beta\cdot\xi)^{-1}$, with $2k$ replaced 
by $\kappa+i\eta$, because equation \eqref{e12} has an operator
$T\epsilon=\int_{B_a}G(x-y,k)q(y)\epsilon(y,k)dy$ with the norm
$||T^2||$ (in the space
$C(B_a)$ of functions with the sup norm) which tends to zero as
$\kappa=2$Re$k\to \infty.$ Consequently, equation \eqref{e12}
can be solved by iterations and the main term in its solution, as 
$|\kappa+i\eta|\to \infty$, $\eta\geq 0$, is the free term
in this equation. The same is true for the Fourier transform
of equation \eqref{e12}, i.e., for equation \eqref{e13}. 

Let us estimate the integral
\be\label{e42}\begin{split} I&=\sup_{\beta\in
S^2}\int_{\R^3}\frac{|\tilde{q}((\kappa+i\eta)\beta-s)
ds}{|((\kappa+i\eta)\beta-s)^2-(\kappa+i\eta)\beta(\kappa+i\eta)\beta-s)|}\\
&\leq c\sup_{\beta\in
S^2}e^{|\eta|a}\int_{\R^3}\frac{ds}{|s^2-((\kappa+i\eta)\beta\cdot
s|[1+(\kappa\beta-s)^2+\eta^2]^{\ell/2}}\\
&:=ce^{|\eta|a}J.\end{split}\ee
Here estimate \eqref{e39} was used.

Let us write the integral $J$ in the spherical coordinates with
$x_3$-axis directed along vector $\beta$, $|s|=r$, $\beta\cdot
s=r\cos\theta :=rt,\quad -1\leq t\leq 1$. Let \be\label{e43}
\gamma:=\kappa^2+\eta^2. \ee Then \be\label{e44}\begin{split}
J&=2\pi\int_0^\infty dr r \int_{-1}^1\frac{dt}{[(r-\kappa
t)^2+\eta^2t^2]^{1/2}(1+\gamma+r^2-2r\kappa t)^{\ell/2}}\\
&:=2\pi \int_0^\infty dr r B(r),
\end{split}\ee
where \be\label{e45}
B:=B(r)=B(r,\kappa,\eta):=\int_{-1}^1\frac{dt}{[(r-\kappa
t)^2+\eta^2t^2]^{1/2}(1+\gamma+r^2-2r\kappa t)^{\ell/2}}. \ee
If
$t\in[-1,1]$, then \be\label{e46} 1+\gamma+r^2-2r\kappa t\geq
1+\gamma^2+r^2-2r\kappa=1+\eta^2+(r-\kappa)^2. \ee 
Thus,
\be\label{e47}\begin{split} B&\leq
\frac{1}{[1+\eta^2+(r-\kappa)^2]^{\ell/2}}\frac{1}{\sqrt{\gamma}}\int_{-1}^1\frac{dt}{[(t-\frac{r\kappa}{\gamma})^2+\frac{\eta^2r^2}{\gamma^2}]^{1/2}}\\
&=\frac{1}{\sqrt{\gamma}[1+\eta^2+(r-\kappa)^2]^{\ell/2}}
\left|\ln\left|\frac{1-\frac{r\kappa}{\gamma}+
\sqrt{(1-\frac{r\kappa}{\gamma})^2+\frac{\eta^2r^2}{\gamma^2}}}
{\sqrt{(1+\frac{r\kappa}{\gamma})^2+\frac{\eta^2r^2}{\gamma^2}}-1-\frac{r\kappa}{\gamma}}
\right|\right|.
\end{split}\ee
Consequently, \be\label{e48} J\leq
\frac{2\pi}{\sqrt{\gamma}}\int_0^\infty\frac{dr
r}{[1+\eta^2+(r-\kappa)^2]^{\ell/2}}\left|\ln\left|\frac{1-\frac{r\kappa}
{\gamma}+\sqrt{(1-\frac{r\kappa}{\gamma})^2+\frac{\eta^2r^2}{\gamma^2}}}
{\sqrt{(1+\frac{r\kappa}{\gamma})^2+\frac{\eta^2r^2}{\gamma^2}}-1-
\frac{r\kappa}{\gamma}}\right|\right|.
\ee
The integral in \eqref{e48} converges: as $r\to
\infty$ the ratio under the logarithm sign tends to 1, and the
factor in front of the logarithm is $O(r^{-(\ell-1)})$ as $r\to
\infty$. Since $\ell>2$, the integral in \eqref{e48} converges.

The modulus of the logarithmic term in \eqref{e48} behaves 
asymptotically, as $r\to 0$, like $|\ln(\frac {r^2\kappa^2}{\gamma^2})|$. 
Thus,  $\lim_{r\to 0}r|\ln(\frac {r^2\kappa^2}{\gamma^2})|=0$
for every fixed $\kappa>0$, and this limit is uniform
with respect to $\kappa$ as $\kappa\to 
\infty$ if $\eta=O(\ln \kappa)$. Therefore, the integrand in \eqref{e48} 
is defined for $r=0$ to be zero by continuity.

As $\gamma=\eta^2+\kappa^2\to \infty$ and $\eta=O(\ln \kappa)$, the 
integrand in \eqref{e48}
tends to zero for every fixed $r\geq 0$, and \eqref{e48} implies 
\be\label{e49} J\leq
o\left(\frac{1}{\sqrt{\gamma}}\right),\quad \gamma\to \infty. 
\ee
Consequently, \eqref{e42} implies 
\be\label{e50} I\leq
cr^{|\eta|a}o\left(\frac{1}{\sqrt{\kappa^2+\eta^2}}\right),\quad
\kappa\to \infty,\ \eta=O(\ln \kappa). \ee 
Therefore, 
\be\label{e51}
\lim_{\kappa\to \infty,\eta=O(\ln\kappa)}I=0. \ee 
This implies estimate \eqref{e40}.\\
\thmref{thm1} is proved. \hfill $\Box$

\appendix
{\it Estimate of the norm of the operator $T^2$.}\\

Let \be\label{A1} Tf:=\int_{B_a}G(x-y,\kappa+i\eta)q(y)f(y)dy. \ee
Assume $q\in H_0^\ell(B_a)$, $\ell>2$, $f\in C(B_a)$. Our goal is to
prove that equation \eqref{e12} can be solved by iterations
for all sufficiently large $\kappa$. 

Consider $T$ as an
operator in $C(B_a)$. One has: 
\be\label{A2}\begin{split}
T^2f&=\int_{B_a}dzG(x-z,\kappa+i\eta)q(z)\int_{B_a}G(z-y,\kappa+i\eta)
q(y)f(y)dy\\
&=\int_{B_a}dyf(y)q(y)\int_{B_a}dzq(z)G(x-z,\kappa+i\eta)G(z-y,\kappa+i\eta).
\end{split}\ee
Let us estimate the integral 
\be\label{A3}\begin{split}
I(x,y):&=\int_{B_a}G(x-z,\kappa+i\eta)G(z-y,\kappa+i\eta)q(z)dz\\
&=\int_{B_a}\frac{e^{i(\kappa+i\eta)[|x-z|-\beta\cdot(x-z)+|z-y|-
\beta\cdot(z-y)]}}{16\pi^2|x-z||z-y|}q(z)dz\\
&=\frac{1}{16\pi^2}\int_{B_a}\frac{e^{i(\kappa+i\eta)[|x-z|+|z-y|-
\beta\cdot(x-y)]}}{|x-z||z-y|}q(z)dz\\
&:=\frac{e^{-i(\kappa+i\eta)\beta\cdot(x-y)}}{16\pi^2}I_1(x,y).\\
\end{split}\ee

Let us use the following coordinates (see \cite{R190}, p.391):
\be\label{A4} z_1=\ell st+\frac{x_1+y_1}{2},\quad
z_2=\ell\sqrt{(s^2-1)(1-t^2)}\cos\psi +\frac{x_2+y_2}{2}, \ee
\be\label{A5} z_3=\ell\sqrt{(s^2-1)(1-t^2)}\sin\psi
+\frac{x_3+y_3}{2}. \ee

The Jacobian $J$ of the ransformation $(z_1,z_2,z_3)\to
(\ell,t,\psi)$ is \be\label{A6} J=\ell^3(s^2-t^2),\ee where
\be\label{A7} \ell=\frac{|x-y|}{2},\quad |x-z|+|z-y|=2\ell s,\quad
|x-z|-|z-y|=2\ell t, \ee \be\label{A8}
|x-z||z-y|=4\ell^2(s^2-t^2),\quad 0\leq \psi<2\pi,\quad t\in[-1,1],\
s\in[1,\infty). \ee 
One has 
\be\label{A9} I_1=\ell\int_a^\infty
e^{2i(\kappa+i\eta)\ell s}Q(s)ds, 
\ee 
where 
\be\label{A10}
Q(s):=Q(s,\ell,\frac{x+y}{2})=\int_0^{2\pi}d\psi\int_{-1}^1dt
q(z(s,t,\psi;\ell,\frac{x+y}{2})), 
\ee 
and the function $Q(s)\in H_0^2(\R^3)$ for any fixed $x,y$. Therefore, an 
integration by parts in \eqref{A9} yields the following estimate: 
\be\label{A11}
|I_1|=O\left(\frac{1}{|\kappa+i\eta|}\right),\quad |\kappa+i\eta|\to 
\infty.
\ee 
From \eqref{A2}, \eqref{A3} and \eqref{A11} one gets:
\be\label{A12} \|T^2\|=O\left(\frac{1}{\sqrt{\gamma}}\right),\qquad
\gamma:=\kappa^2+\eta^2\to \infty. 
\ee 
Therefore, integral equation
\eqref{e12} with $k$ replaced by $\frac{\kappa+i\eta}{2}$, can be
solved by iterations if $\gamma$ is sufficiently large and $\eta\geq
0$. Consequently, integral equation \eqref{e13} can be solved by
iterations. Thus, estimate \eqref{e40} holds if such an estimate
holds for the free term in equation \eqref{e13}, that is, for the function
$\frac{\tilde{q}}{\xi^2-(\kappa+i\eta)\beta\cdot \xi}$. namely,  if
estimate \eqref{e51} holds.

\end{document}